\documentclass{article}

\usepackage[margin=1in]{geometry}

\usepackage{amsmath,amssymb,amsthm}
\usepackage{dsfont}
\usepackage{hyperref}
\usepackage{breakurl}
\usepackage{url}
\usepackage{tikz}
\usepackage{subcaption}
\usepackage{caption}
\usepackage{graphicx}
\usepackage{multicol}
\usepackage{multirow}
\usepackage{amsthm}
\usepackage{lipsum}

\newcommand{\modd}{\,\mathrm{mod}\,}

\newtheorem{theorem}{Theorem}

\newtheorem{assumption}[theorem]{Assumption}
\newtheorem{definition}[theorem]{Definition}
\newtheorem{corollary}[theorem]{Corollary}

\begin{document}

\author{Farhad Farokhi\thanks{CSIRO's Data61 and the University of Melbourne, E-mails: farhad.farokhi@data61.csiro.au; farhad.farokhi@unimelb.edu.au}}
%

\title{\huge Differential Privacy for Evolving Almost-Periodic Datasets with Continual Linear Queries: \\ Application to Energy Data Privacy} 

\date{}

\maketitle

\begin{abstract}
{For evolving datasets with continual reports, the composition rule for differential privacy (DP) dictates that the scale of DP noise must grow linearly with the number of the queries, or that the privacy budget must be split equally between all the queries, so that the privacy budget across all the queries remains bounded and consistent with the privacy guarantees. To avoid this drawback of DP, in this paper, we consider datasets containing almost periodic time series, composed of periodic components and noisy variations on top that are independent across periods. Our interest in these datasets is motivated by that, for reporting on private periodic time series, we do not need to divide the privacy budget across the entire, possibly infinite, horizon. Instead, for periodic time series, we generate DP reports for the first period and report the same DP reports periodically. In practice, however, exactly periodic time series do not exist as the data always contains small variations due to random or uncertain events. For instance, the energy consumption of a household may repeat the same daily pattern  with slight variations due to minor changes to the habits of the individuals. The underlying periodic pattern is a function of the private information of the households, e.g., size of the household, appliances and energy ratings, and average commuting habits. It might be desired to protect the privacy of households by not leaking information about the recurring patterns while the individual daily variations are almost noise-like with little to no privacy concerns (depending on the situation). Motivated by this, we define DP for almost periodic datasets and develop a Laplace mechanism for responding to linear queries. We provide statistical tools for testing the validity of almost periodicity assumption. We use multiple energy datasets containing smart-meter measurements of households to validate almost periodicity assumption. We generate DP aggregate reports and investigate their utility. We show that, by considering the almost periodic nature of energy data, the utility of reporting average consumption per 30 minutes windows over one year can be improved approximately $200$ folds in terms of the standard deviation of the DP noise.} 
\end{abstract}


\section{Introduction}
In many big data applications, such as energy, transportation, and retail, real-time analytic can significantly improve decision making in both short-term (e.g., inventory management) and long-term (e.g., investment opportunities). However, privacy concerns can inhibit access to high-resolution, good-quality data or cause backlash from customers or regulators. Examples of regulatory frameworks to limit access to private data are the General Data Protection Regulation in EU, the California Consumer Privacy Act, and the development of the Data Sharing and Release Bill in Australia.

A natural candidate for treating privacy concerns of individuals is differential privacy (DP)~\cite{dwork2014algorithmic,dwork2006calibrating}, requiring that the statistics of privacy-preserving query responses do not change significantly by changing or removing a single entry of the dataset. The DP framework, however,  often deals with providing privacy-preserving responses to queries based on static datasets and has the highest utility when responding to few aggregate queries. This is because the composition rule for DP dictates that, in order to ensure that the privacy budget remains bounded across all the queries or that it is relatively small to ensure reasonable privacy guarantees, the magnitude of the DP noise must grow linearly with the number of the queries; see, e.g.,~\cite{dwork2014algorithmic} for more information on composition rules for DP. 

More recently, evolving datasets have attracted attention in DP literature with improved privacy and utility guarantees~\cite{dwork2010differential, chan2011private,perrier2019private}. These studies however only consider certain sets of queries (such as, counting queries~\cite{dwork2010differential}) and require that the magnitude of the additive DP noise to grow unboundedly, albeit in a much slower rate. For instance, the best bound so far scales as $\mathcal{O}(\log(t)^{1.5})$, where $t$ denotes the number of observations or queries. Some studies have even proved that magnitude of the additive DP noise might need to grow boundedly for some continual queries~\cite{vadhan2017complexity}.

In this paper, we consider a particular family of datasets containing almost periodic time series, to be able to ensure DP with a finite budget over long, possibly rolling infinite horizons. Almost periodic time series contain a core periodic component and noisy variations on top that are independent across periods; see Figure~\ref{fig:1} for an illustrative example of almost periodic time series. The interest in these datasets is motivated by that, for proving privacy-preserving  reports on periodic time series, we do not need to divide the privacy budget across a very long time, possibly infinite horizon. In fact, for periodic time series, we can generate privacy-preserving reports for the first period and report the same DP reports for each period. In practice, however, there are no periodic time series as variations due to random or uncertain events can slightly change the data. For instance, the energy consumption of a household can repeat the same pattern on a daily basis with slight variations caused by possibly minor changes in the behaviour and the consumption of the individuals. The underlying periodic pattern contains private information of households, e.g., average number of the occupants (i.e., size of the household), appliances and their energy ratings, and commuting habits (i.e., average departure and arrival times). The variations can be caused by daily temperature changes, uncertainties in departure/arrival times of occupants, and seasonal appliance changes. It might be desired to protect the privacy of households by not leaking information regarding the underlying periodic pattern while the individual deviations are almost noise-like with little privacy concerns. However, we can also reduce the amount of the leaked information about the variations on top if they contain private information about the individuals.

In what follows, we define DP for almost periodic datasets. We further develop  Laplace mechanisms for responding to linear queries on almost periodic datasets. We provide statistical tools for testing the validity of the almost periodicity assumption for datasets based on Pearson correlation coefficients. We use Ausgrid dataset~\cite{ausgrid}, containing electricity data for 300 homes with rooftop solar systems, and UMass Dataset~\cite{umass}, containing smart meter measurements of 114 single-family households, to validate almost-periodic assumption of smart-meter measurements. We also use these datasets to generate DP aggregate reports and investigate the utility of such DP policies. 

In summary, in this paper, we make the following contributions:
\begin{itemize}
\item We introduce the concept of almost periodic datasets for privacy-preserving reporting of time series;
\item We define differential privacy for almost periodic datasets and develop Laplace mechanisms for responding to continual linear queries for almost periodic datasets;
\item We provide statistical tools based on correlation coefficients  for testing the validity of the almost-periodic assumption for datasets;
\item We demonstrate the applicability of our results on smart-meter time series of  households in Ausgrid~\cite{ausgrid} and UMass~\cite{umass} datasets. 
\end{itemize}

\subsection{Related Studies}
Homomorphic encryption and secure multi-party computation provide methods for computing energy aggregates without access to individual smart meters~\cite{kursawe2011privacy, rial2011privacy,rastogi2010differentially,shi2011privacy,li2010secure, garcia2010privacy}. These aggregates can be potentially privacy preserving if the number of the households in the aggregate is large enough~\cite{buescher2017two}. However, the use of secure multi-party computation and homomorphic encryption can potentially reduce the computational efficiency of real-time data analytic solutions by adding non-trivial computational and communication performance overhead. Disaggregation and memebership inference attacks can also used to infringe on the privacy of individuals contributing to the aggregate statistics~\cite{buescher2017two,pyrgelis2017knock,pyrgelis2017does}.

Many studies have considered DP in longitudinal datasets containing time series with application to location privacy~\cite{pyrgelis2017knock, pyrgelis2017does, machanavajjhala2008privacy, to2016differentially} and energy data privacy~\cite{zhao2014achieving, acs2011have,danezis2011differentially, rastogi2010differentially}. These studies however focus on datasets containing finite, static time series and often split the privacy budget between the queries. They do not address the underlying problem regarding the need for increasing the privacy budget as more queries are answered. In the past, the lack of a systematic way for setting the privacy budget in longitudinal and using ad hoc methods such as resetting the privacy budget with time has shown to render DP ineffective in practice~\cite{appleprivacyproblems,tang2017privacy}, potentially leaving the individuals vulnerable to privacy breaches and attacks.

Other studies have used renewable generations or batteries to mask the private behaviour in the households~\cite{mclaughlin2011protecting,varodayan2011smart, farokhi2017fisher,giaconi2017smart}. These studies however do not provide the privacy guarantees in the form of DP; their guarantees are often information-theoretic in nature using mutual information, entropy, or Fisher information as measures of private information leakage. These studies also do not address the issue of splitting privacy budget across a long horizon as the complexity of the underlying problem prohibits them from solving the problem for an infinite horizon.

\begin{figure*}[t!]
\begin{tikzpicture}
\node[] at (0,1.8) {Almost periodic time series with period $T=10$};
\node[] at (0,0) {\includegraphics[width=8cm]{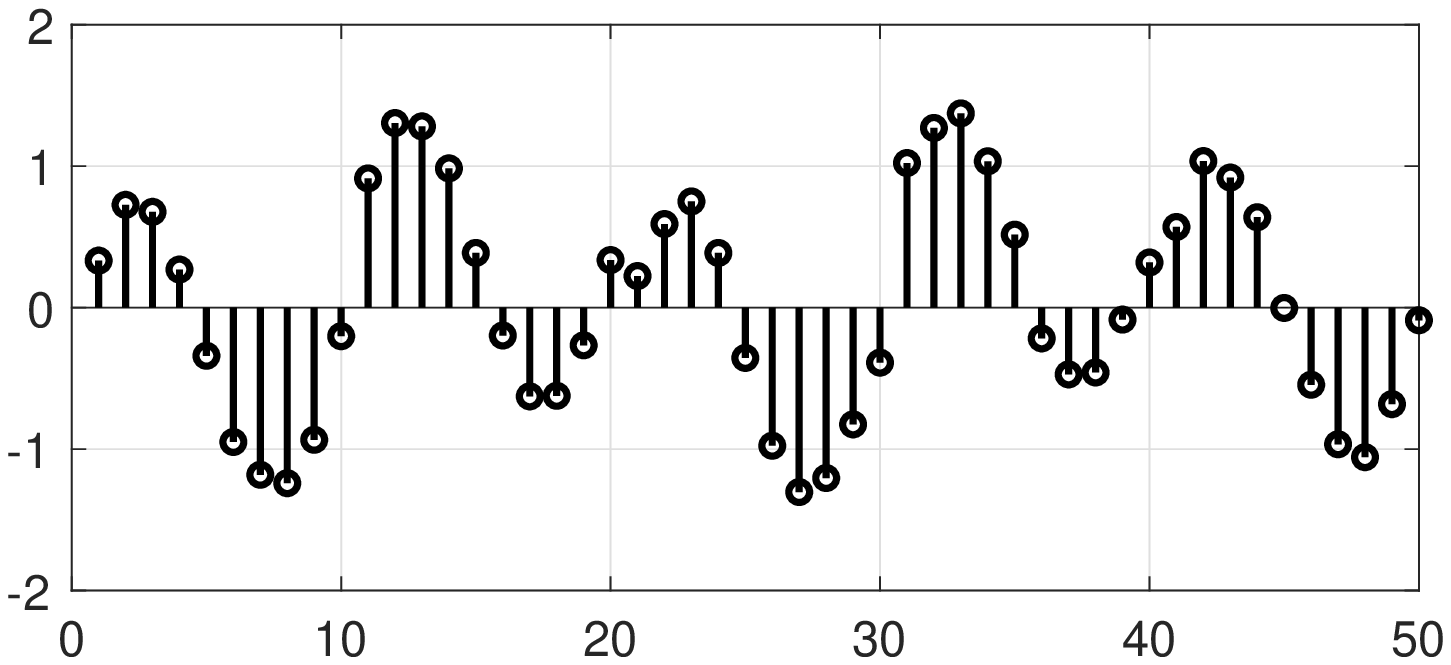}};
\node[] at (0,-1.8) {Time $t$};
\node[rotate=90] at (-3.5,0) {$x(t)$};
\node[] at (9,2) {\includegraphics[width=8cm]{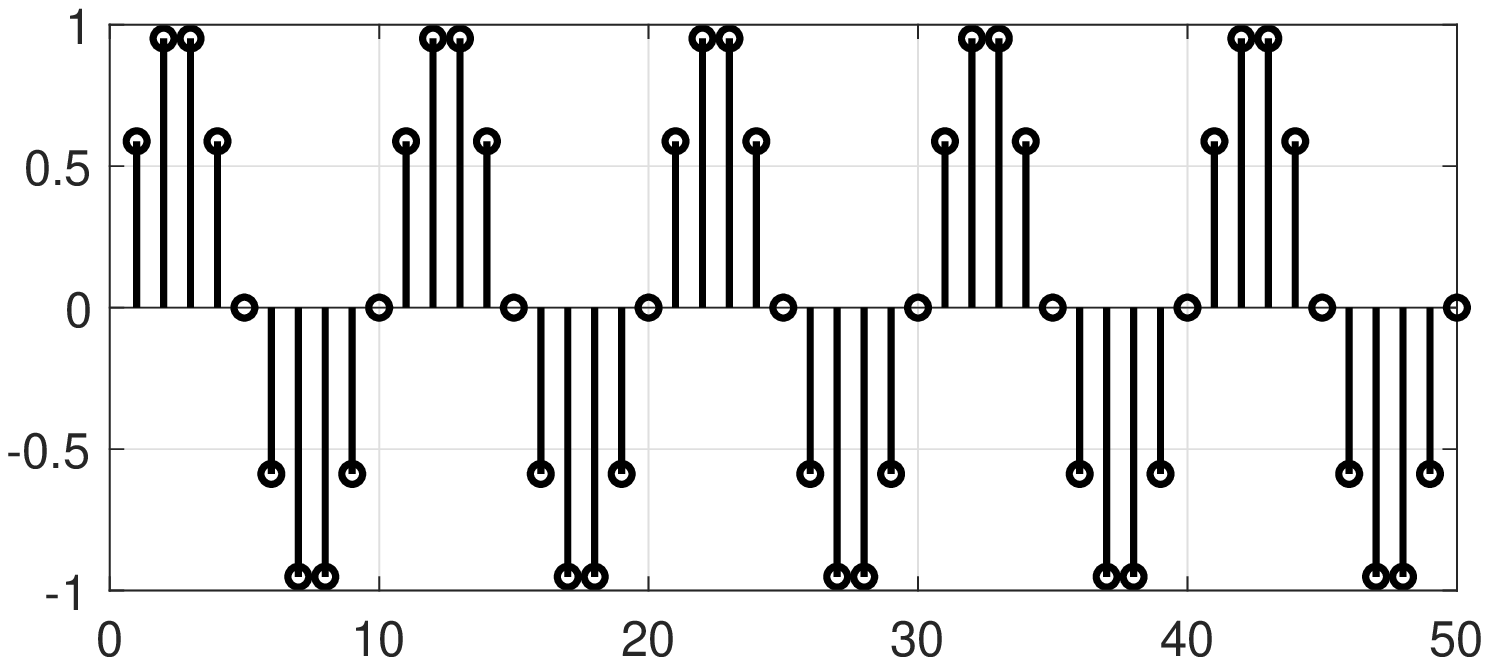}};
\node[] at (9,0.2) {Time $t$};
\node[rotate=90] at (+5.5,2) {$z(t)$};
\node[] at (9,-2) {\includegraphics[width=8cm]{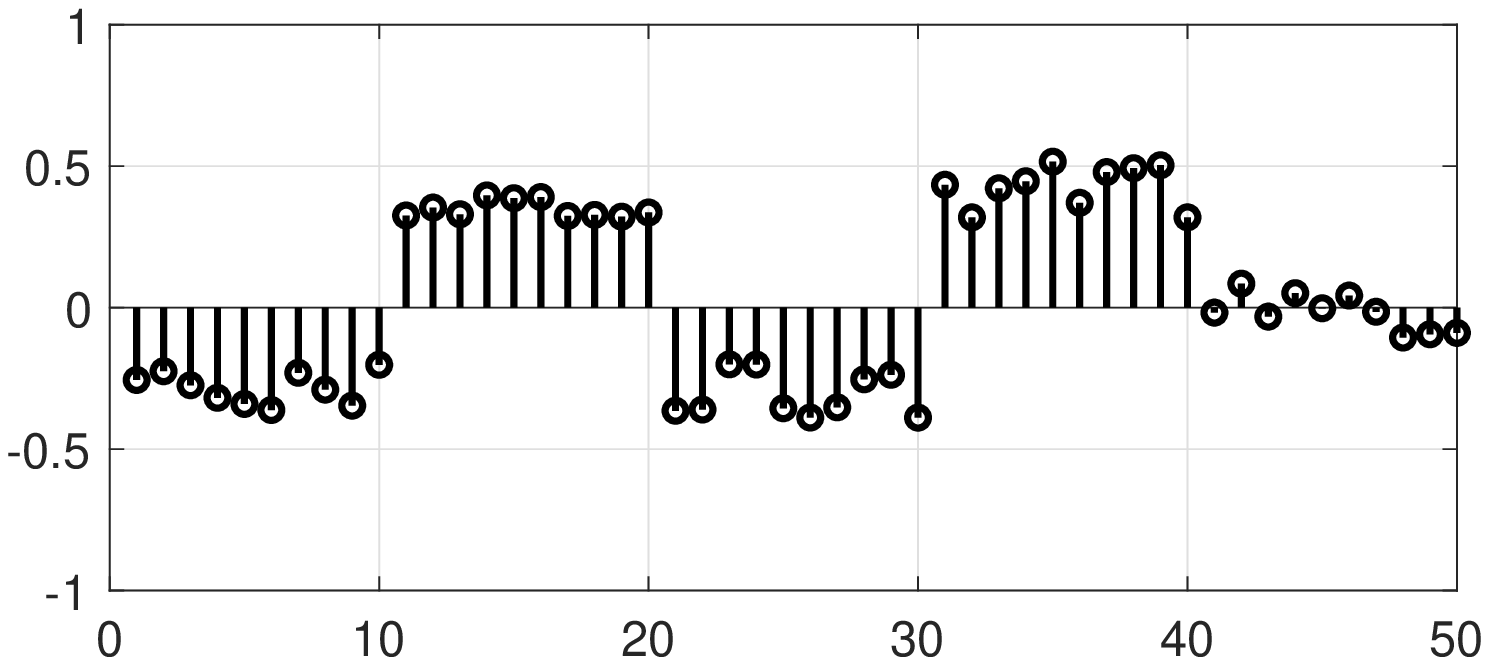}};
\node[] at (9,-3.8) {Time $t$};
\node[rotate=90] at (+5.5,-2) {$w(t)$};
\draw[->,line width=1pt] (+5.1,2) -- (+4.5,2) -- (4.5,0.25);
\draw[->,line width=1pt] (+5.1,-2) -- (+4.5,-2) -- (4.5,-0.25);
\node[circle,draw,minimum width=0.5cm,line width=1pt] at (4.5,0) {};
\draw[->,line width=1pt] (4.25,0) -- (3.5,0);
\node[] at (4.5,0) {\bf +};
\end{tikzpicture}
\caption{ \label{fig:1}
Example of an almost periodic time series with period $T=10$.
}
\end{figure*}

DP for evolving datasets have been more recently studied in several papers~\cite{dwork2010differential, chan2011private,perrier2019private}. As stated earlier, these studies consider specific counting and histogram queries and require an unboundedly growing scale for the DP noise. These studies and their observations are in line with negative results in~\cite{vadhan2017complexity}  stating that the magnitude of the additive DP noise might need to grow bounded for general queries to maintain certain privacy budgets under continual observation.

\subsection{Paper Organization}

The rest of the paper is organized as follows. In Section~\ref{sec:datasets}, we introduce the definition of almost periodic datasets and present differentially-private mechanisms for responding to linear queries on these datasets.  We present methods for testing the validity of the almost periodic assumption in Section~\ref{sec:testing}. We discuss the experimental results in Section~\ref{sec:experiment} and conclude the paper in Section~\ref{sec:conclusions}.

\section{Almost-Periodic Datasets} \label{sec:datasets}
In this paper, we consider evolving dataset, comprised of several time series, in the form of
\begin{align} \label{eqn:dataset}
X(t):=
\begin{bmatrix}
x_1(1) & x_1(2) & \cdots & x_1(t)\\
x_2(1) & x_2(2) & \cdots & x_2(t)\\
\vdots & \vdots & \ddots & \vdots\\
x_n(1) & x_n(2) & \cdots & x_n(t)
\end{bmatrix} \in\mathcal{X}^{n\times t}\subseteq\mathbb{R}^{n\times t},
\end{align}
where $x_i(t)\in\mathcal{X}\subseteq\mathbb{R}$ is the data record associated with individual $i\in\{1,\dots,n\}$ at time instant $t\in\mathbb{N}$.  Note that $t$ can be unbounded and potentially approach infinity. By construction, for any $1\leq k\leq t$, $X(k)$ denotes a matrix extracted by eliminating  the last $t-k$ columns of $X(t)$. An example of such a longitudinal dataset is a dataset containing regular smart-meter reading of $n$ households as in~\cite{ausgrid,umass}.  

\begin{definition}[Almost Periodic] \label{def:almost_periodic} $(x(t))_{t\in\mathbb{N}}$ is almost periodic with period $T\in\mathbb{N}$ if there exists $(z(t))_{t\in\mathbb{N}}$ and $(w(t))_{t\in\mathbb{N}}$ such that 
\begin{enumerate}
\item $\mathbf{z}:=(z(t))_{t\in\mathbb{N}}$, with $z(t)\in\mathcal{Z}\subseteq\mathbb{R}$, is periodic with period $T$, i.e., $z(t+T)=z(t)$ for all $t\in\mathbb{N}$;
\item $(w(t))_{t\in\mathbb{N}}$ with $w(t)=x(t)-z(t)\in\mathcal{W}\subseteq\mathbb{R}$, $\forall t\in\mathbb{N}$, is a time series such that $\mathbf{w}_k:=(w(t))_{(k-1)T+1}^{kT}$ and $\mathbf{w}_\ell=(w(t))_{(\ell-1)T+1}^{\ell T}$ are statistically independent for all $k,\ell\in\mathbb{N}$ if $k\neq \ell$.
\end{enumerate}
There exists an injective relationship, i.e., a one-to-one correspondence, between the almost-periodic time series $(x(t))_{t\in\mathbb{N}}$ and its representation $(\mathbf{z},(\mathbf{w}_k)_{k\in\mathbb{N}})$. We illustrate the injective relationship by $
(x(t))_{t\in\mathbb{N}} \sim (\mathbf{z},(\mathbf{w}_k)_{k\in\mathbb{N}}). $
\end{definition}

Definition~\ref{def:almost_periodic} states that a time series is almost period if it has a periodic component with additive noise-like variations on top of that. The variations are independently  distributed across the periods, e.g., the variations of individuals' habits in a household differ on daily basis and are independent of each other. Note that we do not make any statistical assumptions for the periodic elements of the time series. Furthermore, although we assume that the variations on top of the periodic component are independent across the periods, we do not assume any particular distribution for them.

Figure~\ref{fig:1}, on top of the next page, illustrates an example of an almost periodic time series with period $T=10$. The periodic component in  Figure~\ref{fig:1} is a sampled sinusoidal time series with period $T=10$. The additive noise, that renders the overall signal almost-periodic instead of periodic, is composed of a random constant signal over each period with smaller additive uniform noise at each time instant. 

In this paper, we assume that the rows of the dataset~\eqref{eqn:dataset} are almost periodic with the same period. For instance, in the case of energy data, the same pattern of consumption repeats on a daily basis (hence the assumption of equal frequency for all the households).

\begin{assumption} There exists $T\in\mathbb{N}$ such that $(x_i(t))_{t\in\mathbb{N}}$ for each $i\in\{1,\dots,n\}$ in the evolving dataset~\eqref{eqn:dataset} is almost periodic with period $T\in\mathbb{N}$.
\end{assumption}

To be able to introduce differential privacy for almost-periodic datasets, we need to introduce the notion of neighbouring dataset. 

\begin{definition}[Neighbouring Datasets] \label{def:neighbourhood} Two datasets $X(t),X'(t)\in\mathcal{X}^{n\times t}$ are neighbouring if they only differ for at most one time series, denoted by $i$, and $(x_i(t))_{t\in\mathbb{N}} \sim (\mathbf{z}_i,(\mathbf{w}_{i,k})_{k\in\mathbb{N}})$ and $(x'_i(t))_{t\in\mathbb{N}} \sim (\mathbf{z}'_i,(\mathbf{w}_{i,k})_{k\in\mathbb{N}})$.
\end{definition}

Definition~\ref{def:neighbourhood} implies that two datasets are neighbouring if they only differ for one time series and, for that time series, the periodic components of the time series can only differ, i.e., the variations on top remain similar. The aim of this notion of neighbourhood for datasets, when utilized within the DP definition, is to keep the periodic component of the time series private. The rest is assumed to not contain discernible private  information about the individuals. 

As the dataset evolves, the data custodian reports the following responses:
\begin{align}
Y(t)
=
\begin{bmatrix}
y(1) & y(2) & \cdots & y(t)
\end{bmatrix}
\in\mathcal{Y}^{t}\subseteq\mathbb{R}^t.
\end{align}
Similarly, for any $1\leq k\leq t$, $Y(k)$ denotes a row vector extracted by eliminating the last $t-k$ entries of $Y(t)$. At time instance $t$, to generate the entry $y(t)$, the custodian uses conditional probability density function $p_{y|X}(\cdot|\cdot)$. From now on, we refer to this as the mechanism of the curator.

\begin{definition}[Differential Privacy (DP)] A reporting mechanism by the curator is $\epsilon$-differentially private (DP) over the horizon $\{1,\dots,\tau\}\subseteq\mathbb{N}$ if 
\begin{align} \label{eqn:diffpriv}
\mathbb{P}\{Y(t)\in\mathcal{Y}'|X(t)\}
\leq \exp(\epsilon)
\mathbb{P}\{Y(t)&\in\mathcal{Y}'|X'(t)\},
\forall t\in\{1,\dots,\tau\},
\end{align}
for any Lebesgue measurable set $\mathcal{Y}'\subseteq\mathcal{Y}^t$ and any two neighbouring datasets $X(t),X'(t)$ in the sense of Definition~\ref{def:neighbourhood}.
\end{definition}

In this paper, when responding to queries on the dataset, we use additive privacy-preserving noise. Assume that the curator needs to respond to query $f:\mathcal{X}^n\rightarrow\mathbb{R}$ at every time instant $t\in\mathbb{N}$. 
In the next theorem, we show that, for linear queries (and almost periodic datasets), we can ensure differential privacy for all $t\in\mathbb{N}$ with a finite budget. In what follows, we use the notation $x_{-i}$ to denote $(x_j)_{j\in\{1,\dots,n\}\setminus\{i\}}$. Evidently, by definition, $(x_{-i},x_i)=(x_\ell)_{\ell\in\{1,\dots,n\}}$.

\begin{theorem}  \label{tho:dp}
Let the response to linear query $f$ be given by
\begin{align} \label{eqn:additivenoise}
y(t)=
\begin{cases}
f(x_1(t),\dots,x_n(t))+v(t), & 1\leq t\leq T,\\
y(t\modd T) -f(w_1(t\modd T),\dots,w_n(t\modd T))+f(w_1(t),\dots,w_n(t)), & t>T.
\end{cases}
\end{align}
This response is $\epsilon$-differentially private over the horizon $\{1,\dots,\tau\}\subseteq\mathbb{N}$ for any $\tau\in\mathbb{N}$ if $v(t)$ is a zero-mean Laplace noise with scale $T\Delta f_z/\epsilon$, where
\begin{align} \label{eqn:Deltaf_z}
\Delta f_z:= \sup_{x_{-i}\in\mathcal{Z}^{n-1}}\sup_{x_i,x'_i\in\mathcal{Z}} |f(x_{-i},x_i)-f(x_{-i},x'_i)|.
\end{align}
\end{theorem}

\begin{proof} See Appendix~\ref{proof:tho:dp}.
\end{proof} 

Note that the scale of the additive DP noise in Theorem~\ref{tho:dp} is $T\Delta f_z/\epsilon$, which is \textit{constant and independent of the reporting horizon} $\tau$. This is because we are only interested in  protecting the private information hidden in the periodic components of the dataset. We can extend the notion of neighbouring datasets, and following that extend the definition of DP, to cover information leakage from the additional variations on top as well. This however comes a the cost of requiring larger DP noises. 

\begin{definition}[Strongly Neighbouring Datasets] \label{def:sdp} Two datasets $X(t),X'(t)\in\mathcal{X}^{n\times t}$ are neighbouring  in the strong sense if they only differ for at most one time series, denoted by $i$, and only one of the following statements holds:
\begin{enumerate}
\item $(x_i(t))_{t\in\mathbb{N}} \sim (\mathbf{z}_i,(\mathbf{w}_{i,k})_{k\in\mathbb{N}})$ and $(x'_i(t))_{t\in\mathbb{N}} \sim (\mathbf{z}'_i,(\mathbf{w}_{i,k})_{k\in\mathbb{N}})$;
\item $(x_i(t))_{t\in\mathbb{N}} \sim (\mathbf{z}_i,(\mathbf{w}_{i,k})_{k\in\mathbb{N}})$ and $(x'_i(t))_{t\in\mathbb{N}} \sim (\mathbf{z}_i,(\mathbf{w}'_{i,k})_{k\in\mathbb{N}})$ with at most one $k\in\mathbb{N}$ such that $\mathbf{w}_{i,k}\neq \mathbf{w}'_{i,k}$.
\end{enumerate}
\end{definition}

The first statement in Definition~\ref{def:sdp} is the same as the one in Definition~\ref{def:neighbourhood}. Therefore, any two neighbouring datasets in the sense of Definition~\ref{def:neighbourhood} are also neighbouring in the strong sense. The difference is that the second statement in Definition~\ref{def:sdp} also allows changes in the variations on top besides the periodic components. Now, we can define a stronger notion of DP.

\begin{definition}[Strong Differential Privacy (SDP)] A reporting mechanism by the curator is strongly $\epsilon$-differentially private (SDP) over the horizon $\{1,\dots,\tau\}\subseteq\mathbb{N}$ if~\eqref{eqn:diffpriv} holds for any two datasets $X(t),X'(t)$ that are neighbouring  in the strong sense. 
\end{definition}

In the next corollary, we show that, for linear queries and almost periodic datasets, we can ensure strong differential privacy with bounded privacy budget.

\begin{figure}[t!]
	\centering
	\begin{tikzpicture}
	\node[] at (0,0) {\includegraphics[width=.5\linewidth]{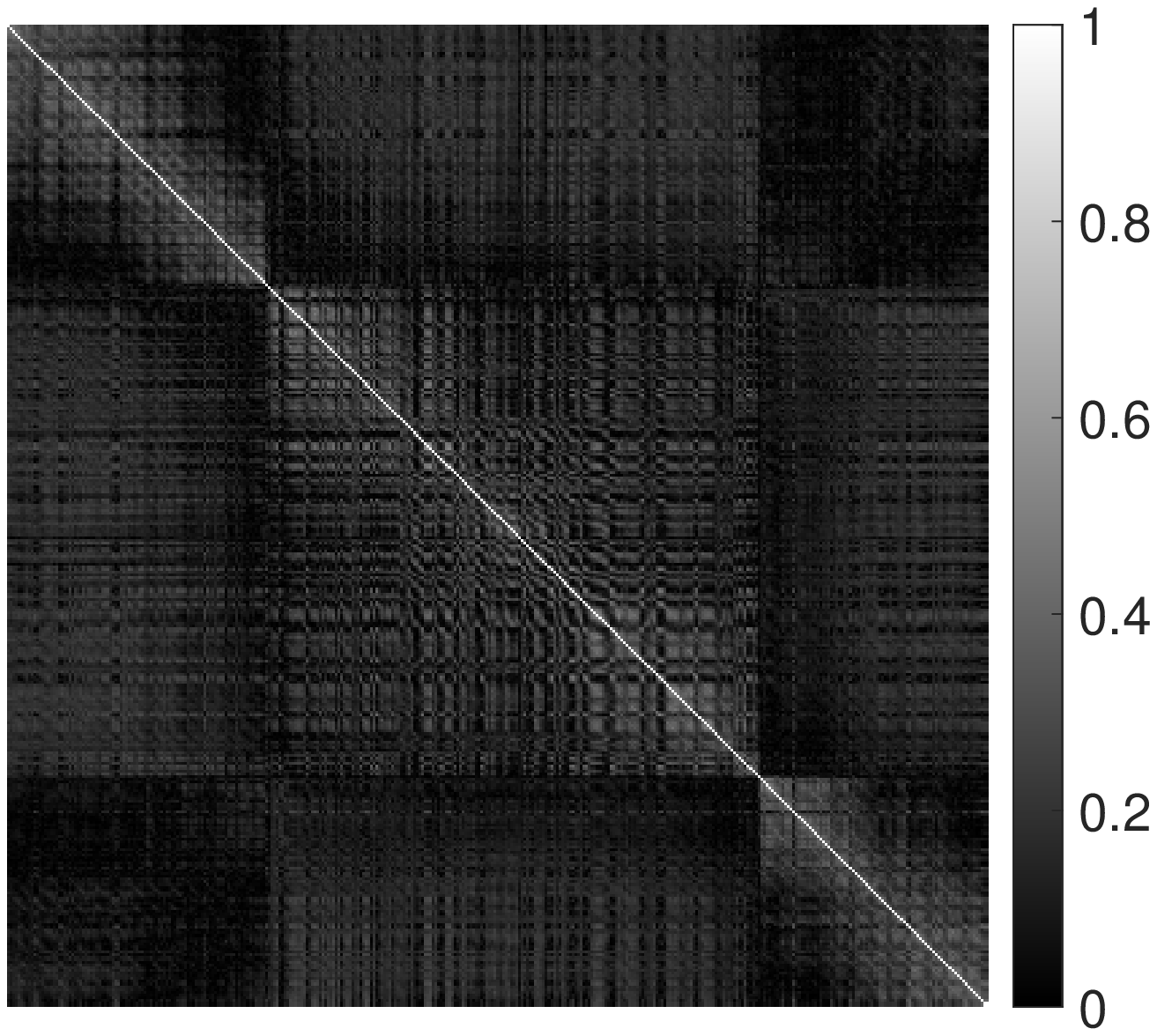}};
	\node[] at (0,-2.8) {$k=1,\dots,365$};
	\node[rotate=90] at (-3.4,0) {$\ell=1,\dots,365$};
	\end{tikzpicture}
	\vspace{-.3in}
	\caption{ Heatmap of the correlation coefficient $\rho_{k\ell}$ versus $k,\ell\in\{1,\dots,365\}$ for the Ausgrid data  with daily period.
		\label{fig:1a}
	}
	\centering
	\begin{tikzpicture}
	\node[] at (0,0) {\includegraphics[width=.5\linewidth]{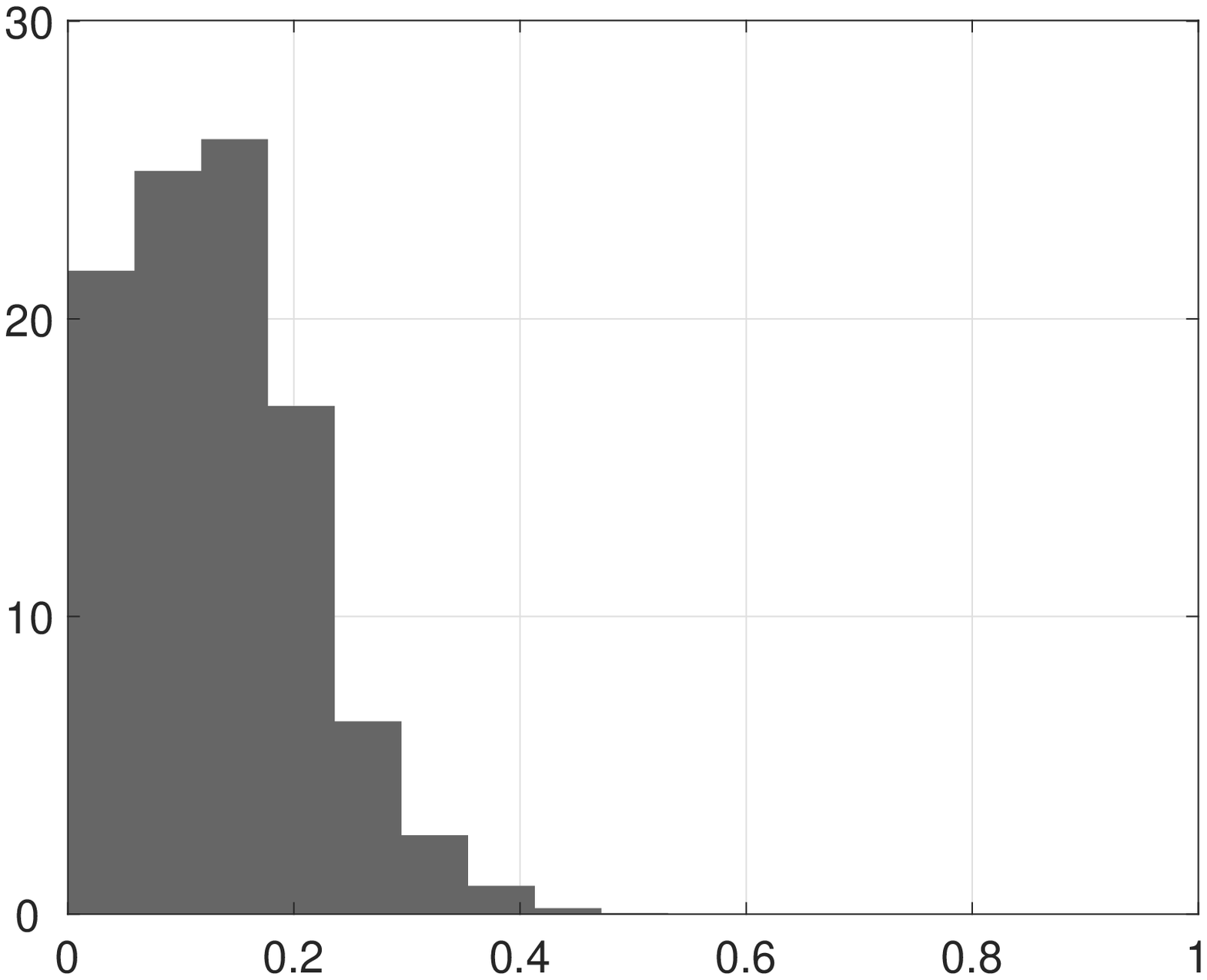}};
	\node[] at (0,-3.2) {magnitude of correlation};
	\node[rotate=90] at (-3.7,0) {Percent};
	\end{tikzpicture}
	\caption{Histogram of the cross-correlation coefficients $(\rho_{k\ell})_{k\neq \ell}$ for the Ausgrid dataset  with daily period. 
		\label{fig:2a}
	}
\end{figure}

\begin{figure}[t!]
	\centering
	\begin{tikzpicture}
	\node[] at (0,0) {\includegraphics[width=.5\linewidth]{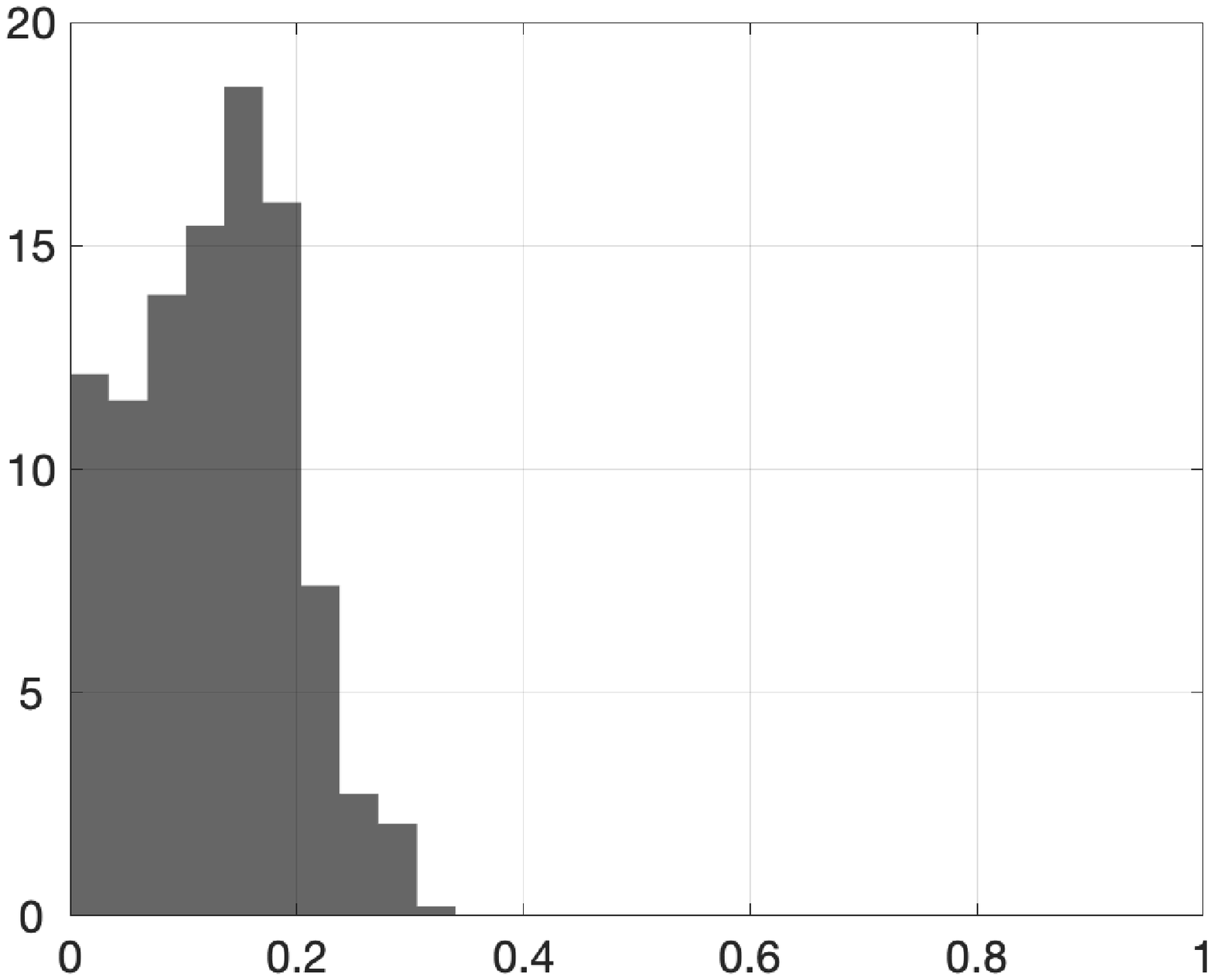}};
	\node[] at (0,-3.2) {magnitude of correlation};
	\node[rotate=90] at (-3.7,0) {Percent};
	\end{tikzpicture}
	\caption{Histogram of the cross-correlation coefficients $(\rho_{k\ell})_{k\neq \ell}$ for the Ausgrid data with weekly period.
		\label{fig:2aa}
	}
\end{figure}

\begin{figure}[t!]
	\centering
	\begin{tikzpicture}
	\node[] at (0,0) {\includegraphics[width=.5\linewidth]{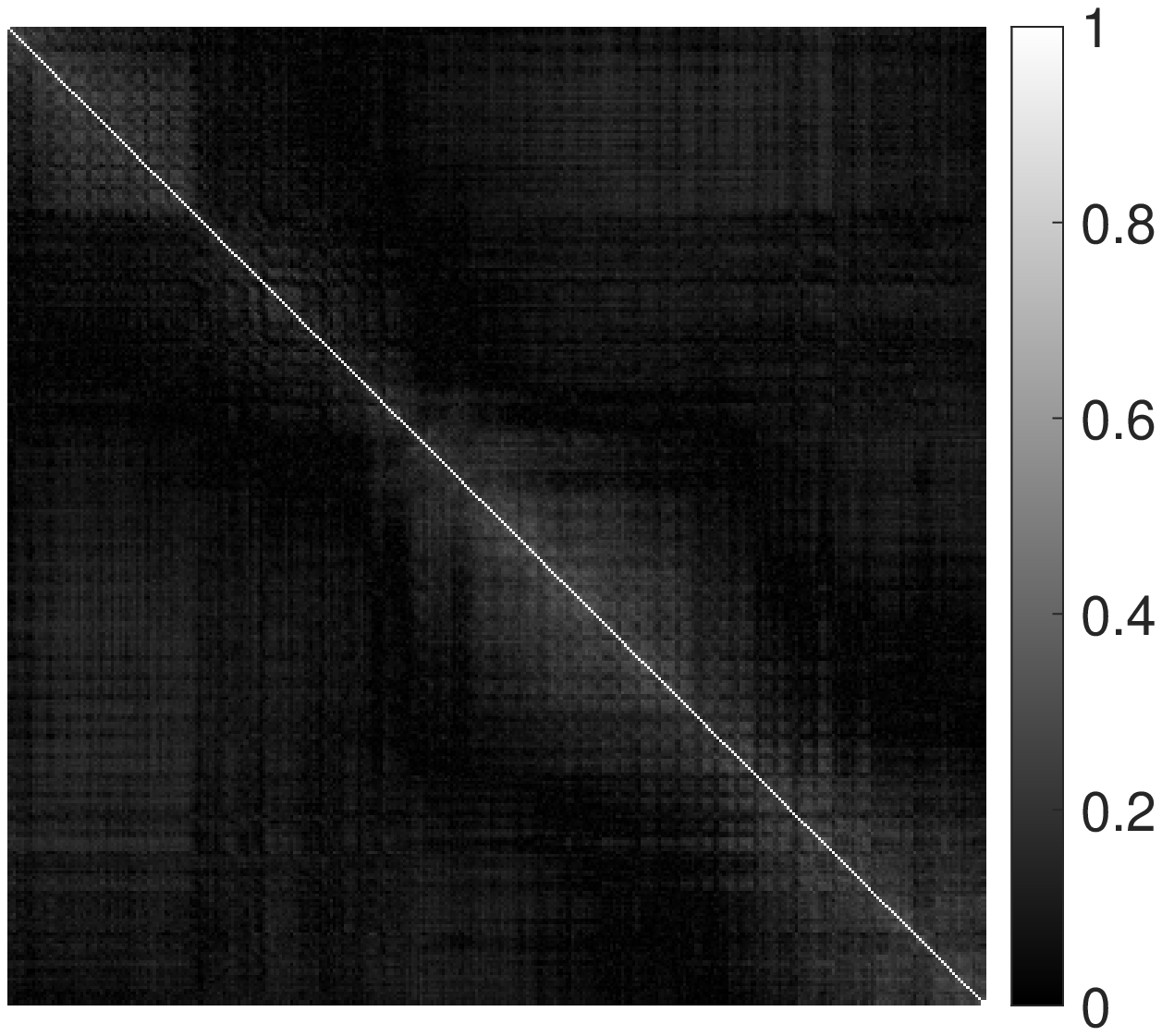}};
	\node[] at (0,-2.8) {$k=1,\dots,349$};
	\node[rotate=90] at (-3.4,0) {$\ell=1,\dots,349$};
	\end{tikzpicture}
	\vspace{-.3in}
	\caption{ 
		Heatmap of the correlation coefficient $\rho_{k\ell}$ versus $k,\ell\in\{1,\dots,349\}$ for the UMass data  with daily period.
		\label{fig:1b}
	}
	\centering
	\begin{tikzpicture}
	\node[] at (0,0) {\includegraphics[width=.5\linewidth]{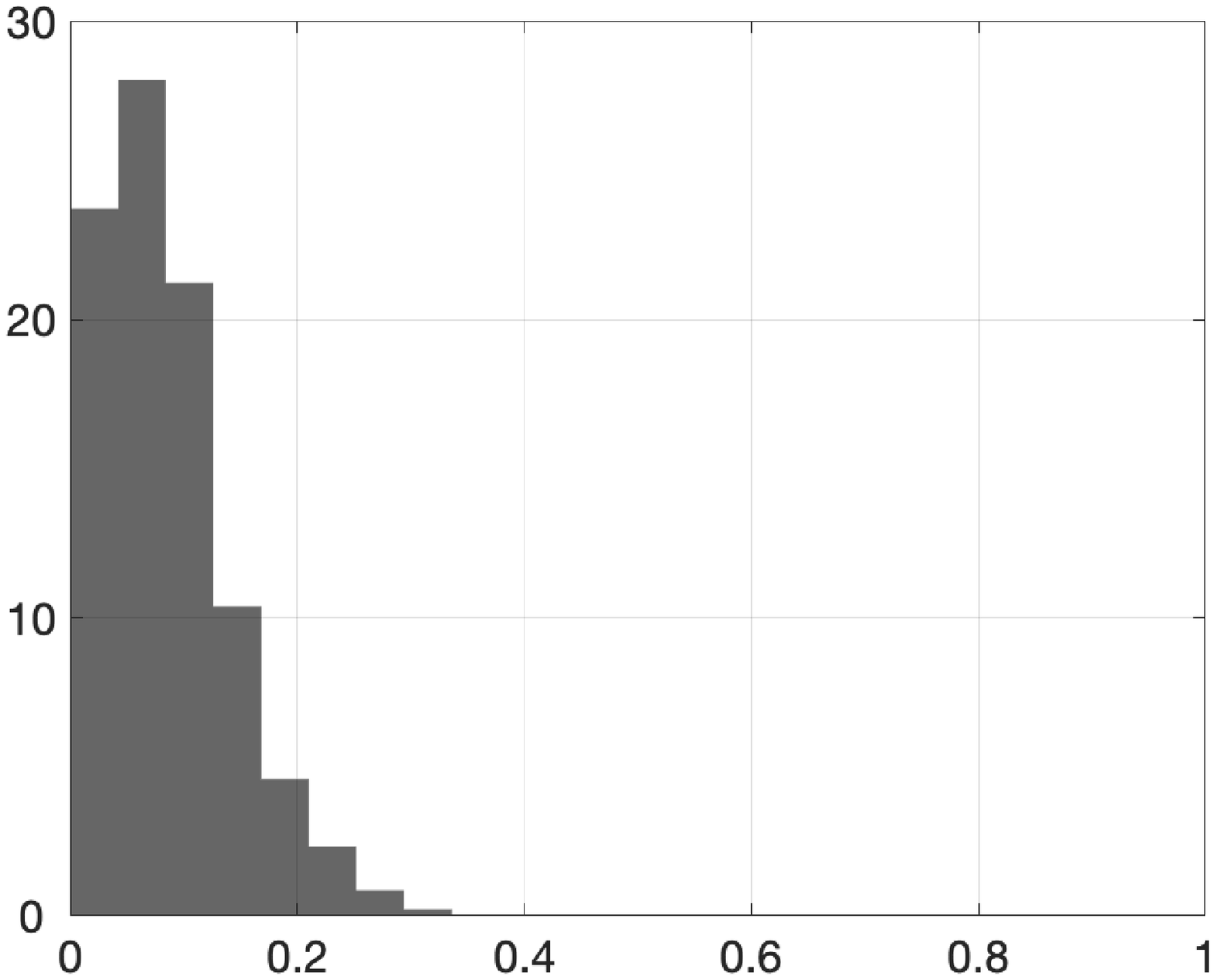}};
	\node[] at (0,-3.2) {magnitude of correlation};
	\node[rotate=90] at (-3.7,0) {Percent};
	\end{tikzpicture}
	\caption{Histogram of the cross-correlation coefficients $(\rho_{k\ell})_{k\neq \ell}$ for the UMass dataset  with daily period. 
		\label{fig:2b}
	}
\end{figure}

\begin{corollary}  \label{cor:sdp}
Let the response to linear query $f$ be given by
\begin{align}
y(t)=
\begin{cases}
f(x_1(t),\dots,x_n(t))+v(t), & 1\leq t\leq T,\\
y(t\modd T)-f(w_1(t\modd T),\dots,w_n(t\modd T))+f(w_1(t),\dots,w_n(t))+v(t), & t>T.
\end{cases}
\end{align}
This response is strongly $\epsilon$-differentially private over the horizon $\{1,\dots,\tau\}\subseteq\mathbb{N}$ for any $\tau\in\mathbb{N}$ if $v(t)$ is a zero-mean Laplace noise with scale 
\begin{align*}
b(t)=
\begin{cases}
T(\Delta f_z+\Delta f_w)/\epsilon, & 1\leq t\leq T,\\
T\Delta f_w/\epsilon, & \mbox{otherwise},
\end{cases}
\end{align*}
 where $\Delta f_z$ is define in~\eqref{eqn:Deltaf_z}
\begin{align*}
\Delta f_w:= \sup_{x_{-i}\in\mathcal{W}^{n-1}}\sup_{x_i,x'_i\in\mathcal{W}} |f(x_{-i},x_i)-f(x_{-i},x'_i)|.
\end{align*}
\end{corollary}

\begin{proof} The proof follows the same line of reasoning as in the proof of Theorem~\ref{tho:dp} in Appendix~\ref{proof:tho:dp}.
\end{proof}


\section{Testing Almost Periodicity} \label{sec:testing}
Assume that we have access to several realizations of of a time series $(x(t))_{t\in\mathbb{N}}\sim (\mathbf{z},(\mathbf{w}_{k})_{k\in\mathbb{N}})$. The realizations are denoted by $(x^j(t))_{t\in\mathbb{N}} \sim (\mathbf{z}^j,(\mathbf{w}_{k}^j)_{k\in\mathbb{N}})$, $j\in\{1,\dots,J\}$, with $J$ denoting the number of the measurements. We can compute the (Pearson) correlation coefficient as
\begin{align*}
\rho_{k\ell}:=
\frac{1}{J-1}\sum_{j=1}^J \frac{1}{\sigma_k\sigma_\ell}(\mathbf{w}_k^{j}-\mu_k)^\top (\mathbf{w}_\ell^{j}-\mu_\ell),
\end{align*}
where
\begin{align*}
\mu_k:=&\frac{1}{J}\sum_{j=1}^J \mathbf{w}_k^{j},
\end{align*}
and
\begin{align*}
\sigma_k:=&\sqrt{\frac{1}{J-1}\sum_{j=1}^J (\mathbf{w}_k^{j}-\mu_k)^\top (\mathbf{w}_k^{j}-\mu_k)}.
\end{align*}
By construction $\rho_{kk}=1$. Furthermore, $0\leq \rho_{k\ell}\leq 1$. If $\rho_{k\ell}$ is small (e.g., $\rho_{k\ell}\ll 1$) for all $k,\ell\in\mathbb{N}$ such that $k\neq \ell$, we can deduce that $\mathbf{w}_k^{j}$ and $\mathbf{w}_\ell^{j}$ are uncorrelated. Although uncorrelated does not always imply independence, many statistical tests for independence, such as the Pearson's chi-squared independence test, rely on it~\cite{pearson1900x}. We analyse the value of the correlation coefficients for energy signals in the subsequent sections to establish their almost periodicity. We use the consumptions of different households in the dataset as realizations to be able to compute the correlation coefficients. 

\section{Experimental Results} \label{sec:experiment}

\subsection{Data}
We use smart meter measurements of of households, made available by the Ausgrid in~\cite{ausgrid} and the UMass in~\cite{umass} to illustrate the results of this paper. The individuals in both these datasets have been de-identified. We have received a waiver from CSIRO's ethics review boards. 

\textbf{Ausgrid Dataset}. The dataset in~\cite{ausgrid} contains electricity data for 300 homes with rooftop solar systems that are measured by a smart meter that, in addition to measuring the usage from the grid,  records the total amount of solar power generated. The measurements are obtained every 30 minutes over 2010-2013. In this paper, we use the data over July 2012 to June 2013.

\textbf{UMass Dataset}. The dataset in~\cite{umass} contains smart meter measurements of 114 single-family households for the period 2014-2016.  In this paper, we use the data in 2016, specifically over the period of Jan 1 to Dec 14,  as it contains energy-consumption measurements in regular one minute intervals. We transform these data to 30 minutes energy consumption readings. 

\subsection{Almost Periodicity}
We first start by investigating the validity of the assumption that the energy data are  almost periodic with period of one day, i.e., $T=48$ as we are dealing with smart meter readings per 30 minutes windows. 

At first, we use the average daily trace of the consumption across the year as the periodic component and the difference of the actual consumptions on any given day from the periodic component as the variations on top. In this case, the period becomes $T=48$ (as we have access to 30 minutes energy consumption readings). Note that parameters, such as temperature and solar availability, can significantly impact the energy consumption from day-to-day and thus the variations on top of the periodic component can be potentially large; however, these variations are caused by natural temperature or solar variations and do not contain private information about the households. These variations can also be statistically independent from one day to another. Note that statistical independence does not exclude time-varying probability density for the daily variations (thus, the variations can differ significantly in the statistical sense across the year).

We start with the Ausgrid dataset. Figure~\ref{fig:1a} illustrates the correlation coefficient $\rho_{k\ell}$ versus $k,\ell\in\{1,\dots,365\}$. We can see there is a high correlation if $k=\ell$, as expected. However, the cross-correlation coefficients $\rho_{k\ell}$, with $k\neq \ell$, are relatively small. Figure~\ref{fig:2a} shows the histogram of the cross-correlation coefficients $(\rho_{k\ell})_{k\neq \ell}$. All the cross-correlation coefficients are all smaller than $0.5$ with at least half the coefficients smaller than or equal to $0.1259$ (which is considerably smaller than one). This points to a strong evidence for almost periodicity of the energy readings in the Ausgrid dataset. 

If we are not convinced by these cross-correlation values, we can consider the average weekly trace of the consumption across the year as the periodic component and the difference of the actual consumptions on any given day from the periodic component as the variations on top. In this case, the period becomes $48\times 7=336$. Figure~\ref{fig:2aa} shows the histogram of the cross-correlation coefficients $(\rho_{k\ell})_{k\neq \ell}$. All the cross-correlation coefficients are all smaller than $0.33$, which is much smaller daily case.

Now, we may consider the UMass dataset. Figure~\ref{fig:1b} illustrates the correlation coefficient $\rho_{k\ell}$ versus $k,\ell\in\{1,\dots,349\}$. Similarly, there is a high correlation if $k=\ell$ and  the cross-correlation coefficients $(\rho_{k\ell})_{k\neq \ell}$ are small. Figure~\ref{fig:2b} shows the histogram of the cross-correlation coefficients $(\rho_{k\ell})_{k\neq \ell}$. All the cross-correlation coefficients are all smaller than $0.35$ with at least half the coefficients smaller than or equal to $0.0746$. These values are significantly smaller than the corresponding values for the Ausgrid dataset. Thus, there seems to be a very strong evidence for almost periodicity of the energy readings in the UMass dataset.

\begin{figure}[t!]
\centering
\begin{tikzpicture}
\node[] at (0,0) {\includegraphics[width=.5\linewidth]{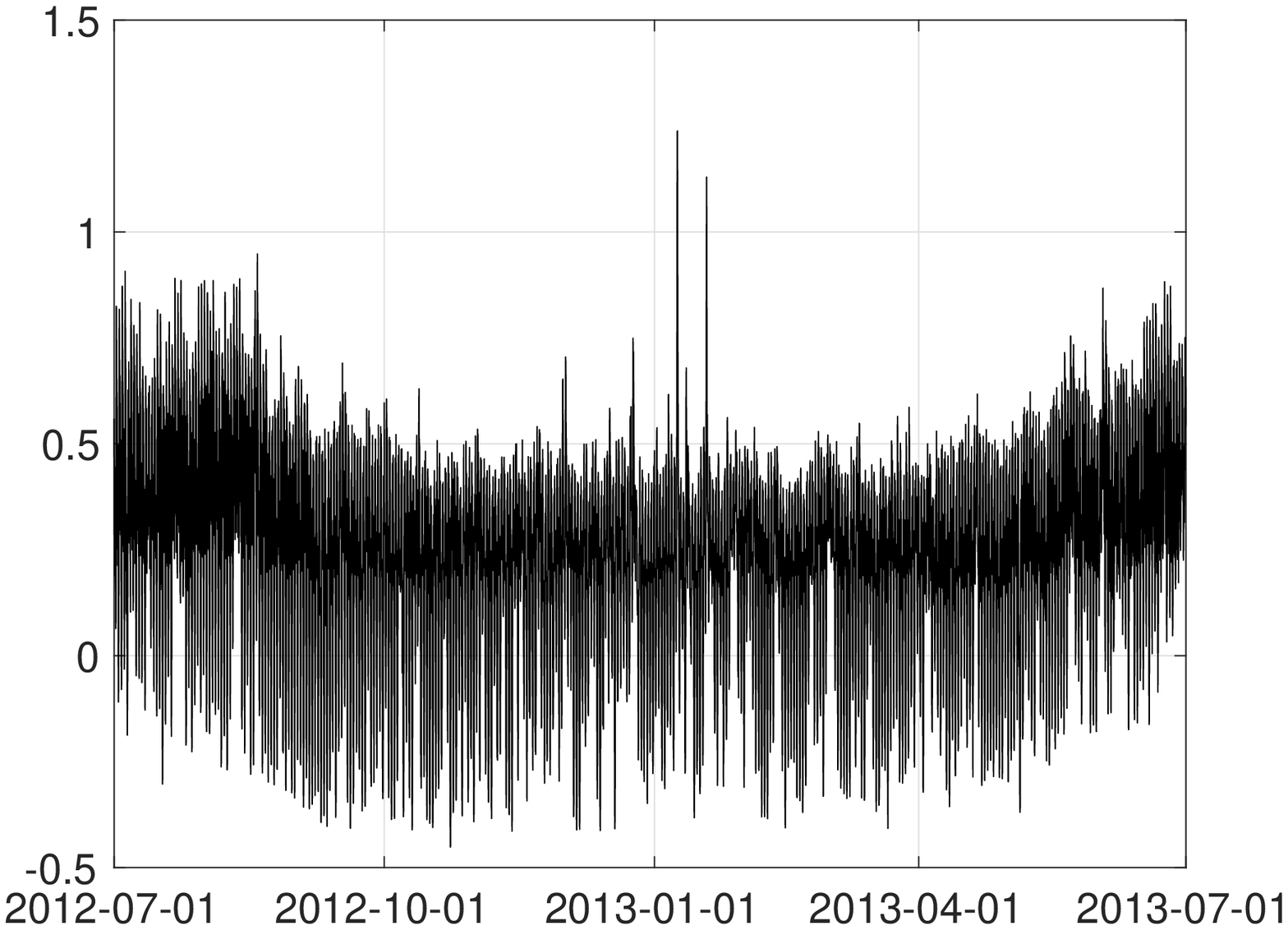}};
\node[] at (0,-3.2) {date $t$};
\node[rotate=90] at (-3.7,0) {$y(t)$};
\end{tikzpicture}
\caption{The outcome of reporting mechanism in Theorem~\ref{tho:dp} with $\epsilon=5$ in the case of the averaging query for the Ausgrid data.
\label{fig:1c}
}
\centering
\begin{tikzpicture}
\node[] at (0,0) {\includegraphics[width=.5\linewidth]{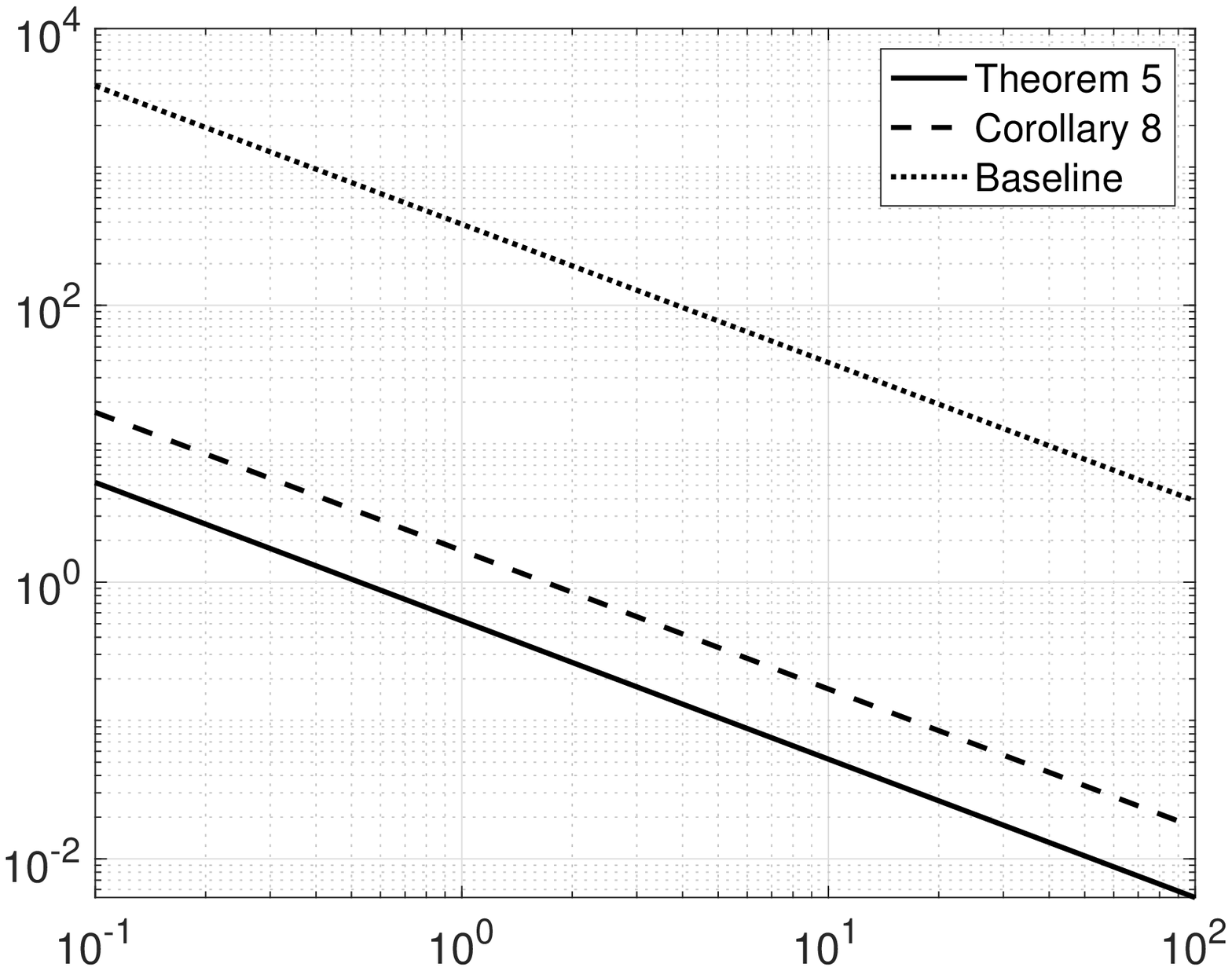}};
\node[] at (0,-3.2) {privacy budget $\epsilon$};
\node[rotate=90] at (-3.9,0) {relative error};
\end{tikzpicture}
\caption{The black line shows the relative error of this reporting mechanism  in Theorem~\ref{tho:dp} for the Ausgrid data. The dashed line shows the relative error for the case where the almost periodicity of the dataset is ignored and the privacy budget is divided equally across the entire horizon.
\label{fig:1d}
}
\end{figure}

\subsection{Utility}
Now, we consider the mechanisms in Theorem~\ref{tho:dp} and Corollary~\ref{cor:sdp} for reporting average consumption \linebreak $f(x_1,\dots,x_n)=(x_1+\cdots+x_n)/n$ of the households in the dataset per 30 minutes. 

We start with the Ausgrid dataset. Figure~\ref{fig:1c} shows the outcome of reporting mechanism in Theorem~\ref{tho:dp} with $\epsilon=5$ for the averaging query $f(x_1,\dots,x_n)$ over Jan 1 to Dec 14, 2016. Let us define the relative error as a measure of utility as
\begin{align*}
\mbox{relative error}:=\frac{\sqrt{\mathbb{E}\{|y(t)-f(x_1(t),\dots,x_n(t))|^2\}}}{\max_t |f(x_1(t),\dots,x_n(t))|}.
\end{align*}
Figure~\ref{fig:1d} illustrates the relative error versus privacy budget $\epsilon$ for the reporting mechanisms in Theorem~\ref{tho:dp} and  Corollary~\ref{cor:sdp} by, respectively, the solid and the dashed lines. As a baseline for comparison, the dotted line illustrates the case where the almost periodicity of the dataset is ignored and the privacy budget is divided equally across the entire horizon of length $16,752=48$ (measurements per day)$\times$ $349$ (days between Jan 1 to Dec 14, 2016). Clearly, the utility of the reports are much better by taking into account the almost-periodicity assumption. 

\begin{figure}[t!]
\centering
\begin{tikzpicture}
\node[] at (0,0) {\includegraphics[width=.5\linewidth]{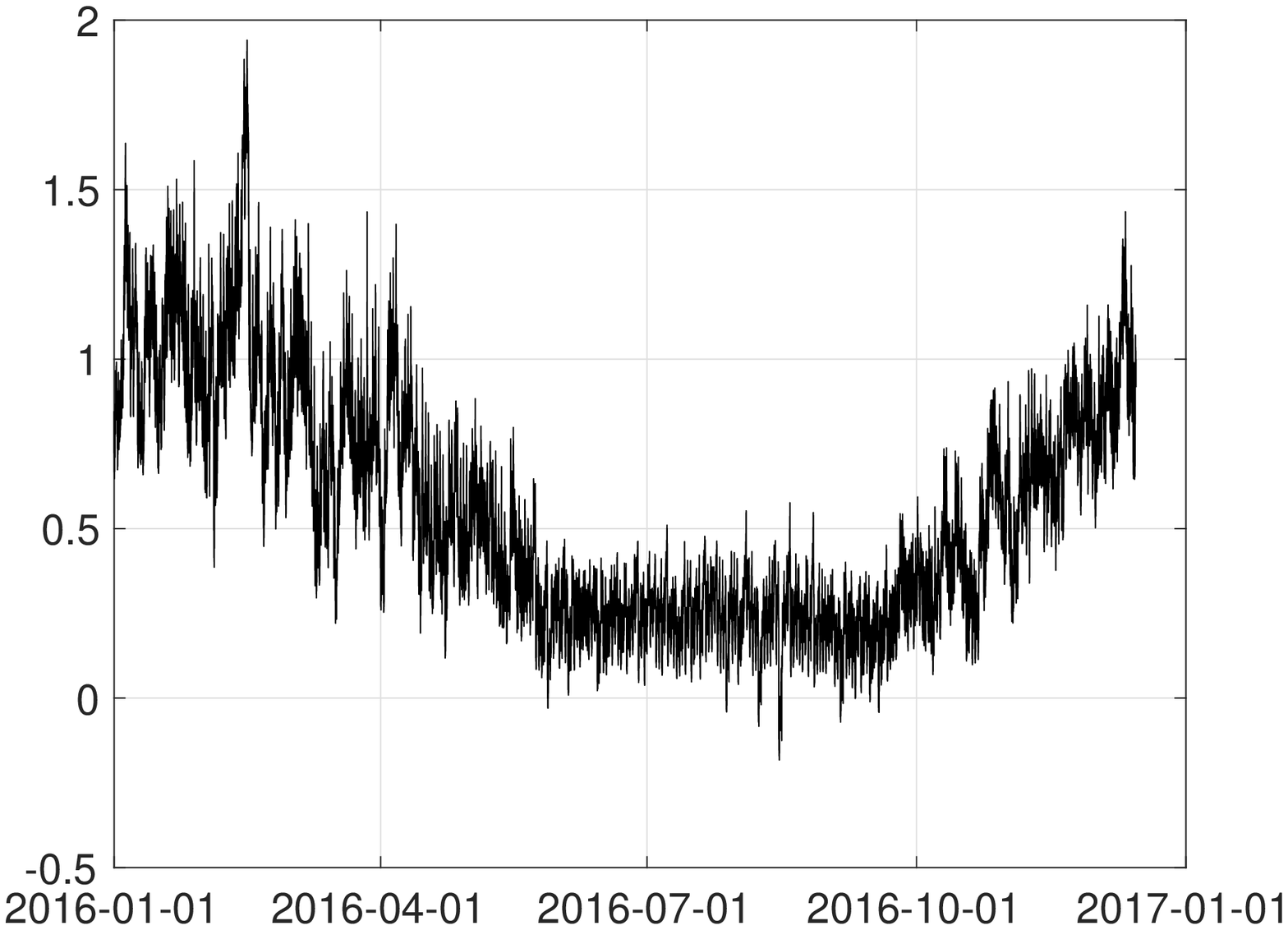}};
\node[] at (0,-3.2) {date $t$};
\node[rotate=90] at (-3.7,0) {$y(t)$};
\end{tikzpicture}
\caption{The outcome of reporting mechanism in Theorem~\ref{tho:dp} with $\epsilon=5$ in the case of the averaging query for the UMass data.
\label{fig:2c}
}
\centering
\begin{tikzpicture}
\node[] at (0,0) {\includegraphics[width=.5\linewidth]{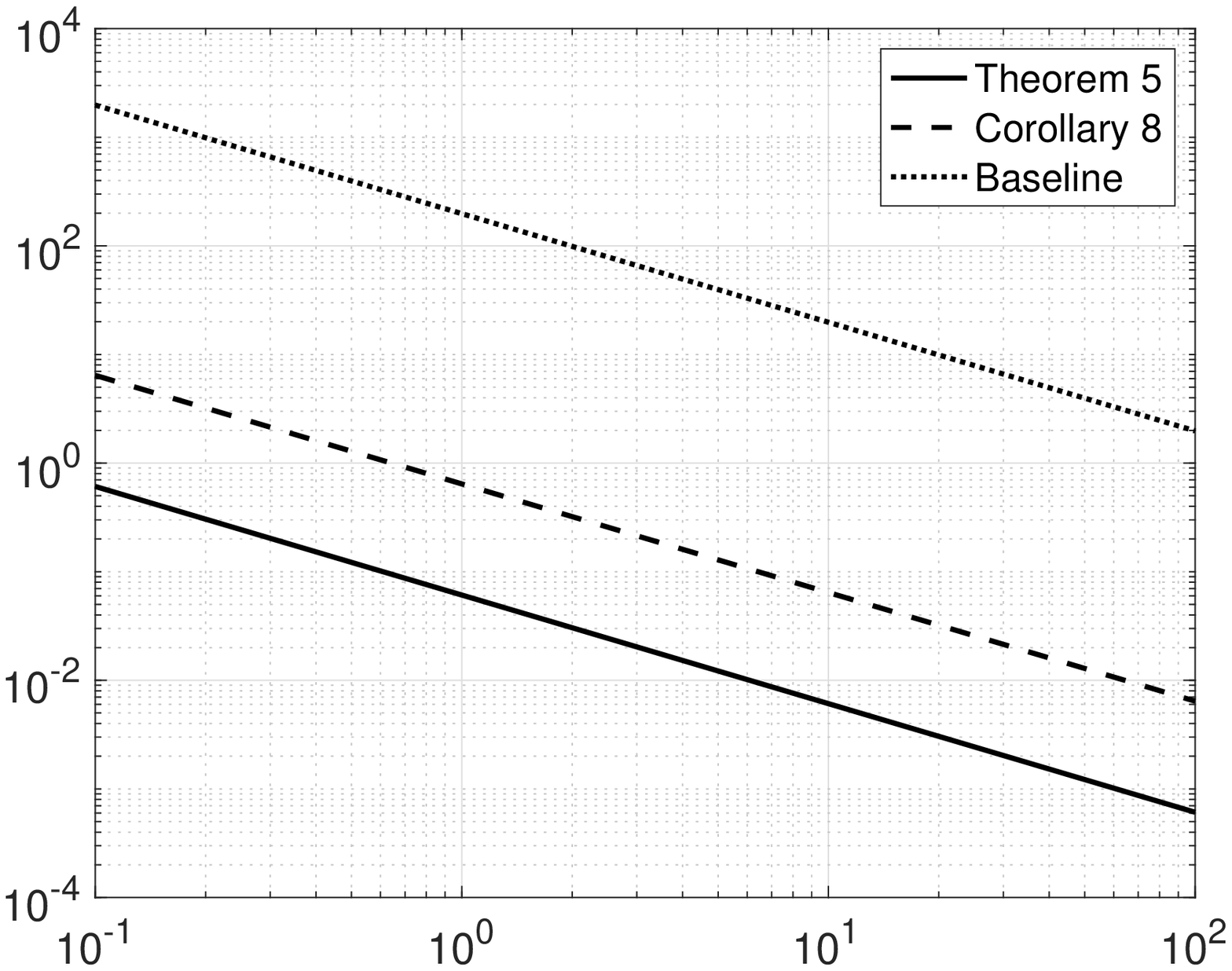}};
\node[] at (0,-3.2) {privacy budget $\epsilon$};
\node[rotate=90] at (-3.9,0) {relative error};
\end{tikzpicture}
\caption{The black line shows the relative error of this reporting mechanism  in Theorem~\ref{tho:dp} for the UMass data. The dashed line shows the relative error for the case where the almost periodicity of the dataset is ignored and the privacy budget is divided equally across the entire horizon.
\label{fig:2d}
}
\end{figure}

Now, we consider the UMass dataset. Figure~\ref{fig:2c} shows the outcome of reporting mechanism in Theorem~\ref{tho:dp} with $\epsilon=5$ for the averaging query. Similarly, the solid and dashed lines in Figure~\ref{fig:2d}  illustrates the relative error of the reporting mechanisms in Theorem~\ref{tho:dp} and Corollary~\ref{cor:sdp}, respectively. Furthermore, the dotted line shows the relative error for the case where the almost periodicity of the dataset is ignored and the privacy budget is divided equally across the entire horizon of length $17,520=48$ (measurements per day)$\times$ $365$ (days between Jan 1 to Dec 14, 2016). Again, as expected, the utility of the reports are much better by taking into account the almost periodicity condition. 

\section{Conclusions and Future Work} \label{sec:conclusions}
We considered the privacy of almost periodic-time datasets, such as energy data.
We defined DP for almost periodic datasets and developed a Laplace mechanism for responding to linear queries. We illustrated the utility of these privacy-preserving mechanisms on real energy data. Future work can focus on generalizing the results to larger family of queries, i.e., non-linear queries.

\bibliographystyle{plain}
\bibliography{citation}

\appendix

\section{Proof of Theorem~\ref{tho:dp}}
\label{proof:tho:dp}
Note that
\begin{align*}
\frac{p((y(t))_{t=1}^T|X(t))}{p((y(t))_{t=1}^T|X'(t))}&=\prod_{t=1}^T
\frac{\exp(-|y(t)-f(x_1(t),\dots,x_n(t)))|/b)}{\exp(-|y(t)-f(x'_1(t),\dots,x'_n(t)))|/b)}\\
&=
\prod_{t=1}^T\exp(-|y(t)-f(x_1(t),\dots,x_n(t)))|/b+|y(t)-f(x'_1(t),\dots,x'_n(t)))|/b),
\end{align*}
where
$b$ denotes the scale of the Laplace noise. As a result, we get
\begin{align*}
\frac{p((y(t))_{t=1}^T|X(t))}{p((y(t))_{t=1}^T|X'(t))}
& \leq \prod_{t=1}^T\exp(|f(x_1(t),\dots,x_n(t))-f(x'_1(t),\dots,x'_n(t)))|/b)\\
&\leq \prod_{t=1}^T\exp(\Delta f_z/b)\\
&=\exp(T \Delta f_z/b).
\end{align*}
We get
\begin{align*}
\frac{p((y(t))_{t\in\mathbb{N}}|X(t))}{p((y(t))_{t\in\mathbb{N}}|X'(t))}
&=\frac{p((y(t))_{t\in\mathbb{N}\setminus\{1,\dots,T\}}|X(t),(y(t))_{t=1}^T)}{p((y(t))_{t\in\mathbb{N}\setminus\{1,\dots,T\}}|X'(t),(y(t))_{t=1}^T)} \frac{p((y(t))_{t=1}^T|X(t))}{p((y(t))_{t=1}^T|X'(t))}\\
&=\frac{p((y(t))_{t=1}^T|X(t))}{p((y(t))_{t=1}^T|X'(t))}\\
&\leq \exp(T \Delta f_z/b),
\end{align*}
where the first inequality follows from that $(y(t))_{t\in\mathbb{N}\setminus\{1,\dots,T\}}$ is deterministically determined by  
\begin{align*}
y(t)=y(t\modd T)&-f(w_1(t\modd T),\dots,w_n(t\modd T))+f(w_1(t),\dots,w_n(t)).
\end{align*}
Selecting $b=T\Delta f_z/\epsilon$ ensures differential privacy with budget $\epsilon$.

\end{document}